\documentclass[12pt]{article}
\usepackage{amsmath}
\usepackage{amssymb}
\usepackage[onehalfspacing]{setspace}

\newtheorem{theorem}{Theorem}

\newtheorem{definition}{Definition}
\newtheorem{example}{Example}

\newtheorem{lemma}{Lemma}

\newenvironment{proof}[1][Proof]{\noindent \textbf{#1.} }{\  \rule{0.5em}{0.5em}}

\begin{document}

\author{Mario V\'{a}zquez Corte\thanks{%
Department of Economics, ITAM}}
\title{A Model of Choice with Minimal Compromise\thanks{%
This work draws from my work under the supervision of Levent \"{U}lk\"{u}, who erroneously appeared as a coauthor in a previous draft. 
}}
\date{This version: July 2016}
\maketitle

\begin{abstract}
I formulate and characterize the following two-stage choice behavior. The
decision maker is endowed with two preferences. She shortlists all maximal
alternatives according to the first preference. If the first preference is
decisive, in the sense that it shortlists a unique alternative, then that
alternative is the choice. If multiple alternatives are shortlisted, then,
in a second stage, the second preference vetoes its minimal alternative in
the shortlist, and the remaining members of the shortlist form the choice
set. Only the final choice set is observable. I assume that the first
preference is a weak order and the second is a linear order. Hence the
shortlist is fully rationalizable but one of its members can drop out in the
second stage, leading to bounded rational behavior. Given the asymmetric
roles played by the underlying binary relations, the consequent behavior
exhibits a minimal compromise between two preferences.
To our knowledge it is the first Choice function that satisfies Sen's $\beta$ axiom of choice, but not $\alpha$.

\noindent \textbf{J.E.L. codes:} D0.\newline

{\footnotesize \ }\vfill 
\end{abstract}

\bigskip

\section{Introduction}

The standard model of rational choice centers around a decision maker (DM)
who maximizes a given preference in every menu. Experimental evidence and
field data contain robust deviations from this model. The accumulation of
such evidence has created an interest in developing new models of bounded
rationality which rely on a richer set of psychological variables.

A particularly prominent idea which has been explored in this literature is
that the DM might use multiple preferences in making choices. In the
presence of multiple preferences, any conflict which may arise between
preferences need to be resolved before making choices.\footnote{%
Of course, conflicts between criteria need not be resolved and the DM may
choose \emph{any} alternative which is the best according to \emph{some}
preference. The resulting behavior is characterized by the famous Path
Independence axiom. See, for instance, Moulin (1985).} Recent work has
studied various ways of resolving such conflicts, mainly by explicitly
attributing different roles to different preferences. In Manzini and
Mariotti (2007) and Bajraj and \"{U}lk\"{u} (2015), for example, the DM uses
one preference to identify a shortlist of viable alternatives and a second
preference to choose from the shortlist\footnote{%
This is further discussed in Horan (2016), and Garc\'{\i}a-Sanz and Alcatud
(2015). Both analyze the two-stage procedure inspired by Manzini and
Mariotti (2007). }.

In this work I will study a model which, similarly, features a compromise
between two preferences. In my model, the DM will choose to maximize a
preference, with the proviso that, in case multiple alternatives are
maximal, a second preference will be able to veto an alternative. To be
precise, the model works as follows. The DM is endowed with a weak order (a
utility) and a linear order (a utility where no distinct alternatives are
indifferent.) The DM first shortlists all best alternatives in the weak
order. If a unique alternative is shortlisted, it is chosen. Otherwise, in a
second stage, she eliminates from the shortlist the worst alternative
according to the linear order. The remaining alternatives form the choice
set.

In view of the very asymmetric roles played by the two underlying
preferences, this model features an idea of a \emph{minimal} compromise.
Note that the linear order plays no role if the weak order is decisive in
the first round. Only if the weak order shortlists exactly two candidates,
does the linear order make the choice in the second round. If more than two
are shortlisted, the linear order can only veto one of them. Hence the
departure from the maximization of the first stage preference as a result of
a conflict with the second preference is, in some sense, minimal.

In terms of behavior, this model can explain violations of Sen's $\alpha $
axiom, however it necessarily satisfies various other rationality axioms
such as $\gamma $, $\beta $ and No Binary Cycles (NBC). My main result is a
characterization of this model using five novel conditions.

My model is largely inspired by Manzini and Mariotti (2007). They study a
two-stage choice procedure which depends on two asymmetric binary relations.
In the first stage, the DM forms a shortlist consisting of all maximal
alternatives according to the first binary relation. In the second stage she
chooses from the shortlist using the second binary relation. They show that
this two-stage procedure explains cyclical behavior whereby $x$ is chosen
over $y$, $y$ over $z$ and $z$ over $x$. The main difference between the
present work and Manzini and Mariotti (2007) is that I study a choice
correspondence, while they characterize a choice function, which is a
restrictive form of a choice correspondence which selects a unique
alternative in every menu. Furthermore, in my model the role of the second
binary relation is different. Instead of choosing its best-preferred
alternative from the shortlist, it vetoes the choice of its least preferred
shortlisted alternative. I should note that, as does the related two-stage
model of Bajraj and \"{U}lk\"{u} (2015), my model fails to account for
cyclical behavior. Instead, my model can explain violations of $\alpha $,
whereby an alternative drops out of the choice set in a smaller menu which
contains it.\\

\section{Interpretation}

The Choice with Minimal Compromise can be interpreted as a two different agents choosing over a menu. Imagine you are in a restaurant with your significant other and decide to order pizza. You read the menu and enumerate your favorite options: Hawaiian, supreme, veggie and 4 cheeses. Your significant other then says he doesn't want to eat supreme. So you compromise and decide to choose a pizza from the remaining three alternatives.\\ 

Formally, the firs agent is represented by the preference relation $R$ shortlist her best alternatives. Then asks the second agent, to take out his least favorite option from the shortlist encoded by the preference relation $L$. then, she proceeds to choose any item from the remaining shortlist.

\section{Model}

I consider a standard choice environment. Let $X$ be a finite set of
alternatives. A binary relation $R$ on $X$ is (1) \emph{complete} if for all 
$x,y\in X$, either $xRy$ or $yRx$, (2) \emph{transitive} if for all $%
x,y,z\in X$, if $xRy$ and $yRz$, then $xRz$, and (3) \emph{antisymmetric} if
for all $x,y\in X$, if $xRy$ and $yRx$, then $x=y$. A \emph{weak order} is a
complete and transitive binary relation. A \emph{linear order} is an
antisymmetric weak order. We will typically refer to weak orders by $R$ and
linear orders by $L$. For any weak order $R$, I will denote by $I$ and $P$ the symmetric and asymmetric parts of $R$, respectively: $xIy\Leftrightarrow xRy$ and $yRx$; and $xPy\Leftrightarrow xRy$ and $\neg (yRx)$.

A \emph{menu} is any nonempty subset of $X$. $2^{X}=\{A\subseteq
X:A\not=\varnothing \}$ denotes the set of menus. If $R$ is a weak order,
let $\max (A,R)$ denote the set of maximal alternatives in $A$ according to $%
R$, in other words, $\max (A,R)=\{x\in A:xRy$ for all $y\in A\}$. If $L$ is
a linear order, then $\max (A,L)$ is a singleton, as is $\min (A,L)=\{x\in
A:yLx$ for all $y\in A\}$. In this case, I will refer to the unique
alternative in $\max (A,L)$ (resp. $\min (A,L)$) by $\max (A,L)$ (resp. $%
\min (A,L)$) as well.

A\ \emph{choice correspondence} is a map $c:2^{X}\rightarrow 2^{X}$
satisfying $c(A)\subseteq A$ for every menu $A$. If $c(A)=\{x\}$ for some
menu $A$ and some alternative $x\in A$, I will say that $c$ is \emph{decisive%
} at $A$. A choice correspondence $c$ is \emph{rational} if there exists a
weak order $R$ such that $c(A)=\max (A,R)$ for every menu $A$. A choice
correspondence $c$ satisfies the \emph{weak axiom of revealed preference}
(WARP) if for every $x,y,A$ and $B$, if $x,y\in A\cap B$, $x\in c(A)$ and $%
y\in c(B)$, then $x\in c(B)$ as well. It is well known that WARP is a
necessary and sufficient condition for the rationality of choice
correspondences. (See for instance Moulin, 1985.)

I can now define the class of choice correspondences of interest in this
work.\bigskip

\begin{definition}
\emph{A choice correspondence }$c$\emph{\ admits a }minimal compromise\emph{%
\ representation if there exist a weak order }$R$\emph{\ and a linear order }%
$L$\emph{\ such that for every menu }$A$%
\begin{equation*}
c(A)=\left\{ 
\begin{tabular}{ll}
$\max (A,R)$ & \emph{if }$\max (A,R)$\emph{\ is a singleton,} \\ 
$\max (A,R)\backslash \min (\max (A,R),L)$ & \emph{otherwise.}%
\end{tabular}%
\right.
\end{equation*}%
\emph{If this is the case, I will call }$c$\emph{\ an MC\ choice
correspondence.}
\end{definition}

Hence a choice correspondence $c$ with a minimal compromise representation
operates in two stages. In the first stage maximal alternatives according to 
$R$ are shortlisted. In the second stage the choice is made from the
shortlist using $L$ as follows. If the shortlist contains only one
alternative, then it is the choice, and $L$ has no role to play. If the
shortlist contains multiple alternatives, however, the alternative which is
the worst according to $L$ is eliminated. All remaining alternatives form
the choice set. Hence the compromise between $R$ and $L$ is minimal, in the
sense that $L$ can veto only one alternative from the shortlist, if the
shortlist contains two or more alternatives.

MC choice correspondences can usefully explain failures of WARP. Consider
the following two axioms:\bigskip

\noindent $\alpha $: If $x\in c(A)$ and $x\in B\subset A$, then $x\in c(B)$%
.\bigskip

\noindent $\beta $: If $x,y\in A\subset B,$ $x,y\in c(A),$ and $y\in c(B),$
then $x\in c(B)$ as well.\bigskip

\noindent \textbf{Fact 1:} (Moulin, 1985) A\ choice correspondence satisfies
WARP if and only if it satisfies $\alpha $ and $\beta $.\bigskip

My first result indicates that MC choice correspondences can fail $\alpha $,
but they have to satisfy $\beta $.

\begin{lemma}
Let $c$ admit a minimal compromise representation. Then $c$ satisfies $\beta 
$. However $c$ may fail $\alpha $.
\end{lemma}

\begin{proof}
Let $c$ admit a minimal compromise representation. Fix $x,y\in $ $%
A,B\subseteq X$. Suppose $x,y\in c(A)$, $A\subset B$ and $y\in c(B).$ I need
to show that $x\in c(B).$ If $x=y$ the condition is automatically satisfied
since $y\in c(B).$ Suppose $x\neq y.$ Since $y\in c(B)$, $y\in \max (B,R)$
and therefore $yRb$ for all $b\in B.$ The fact that $x\in c(A)$ implies that 
$x\in \max (A,R)$, so $xRy$ since $y$ belongs to the set $A.$ Now
transitivity of $R$ gives $xRb$ for all $b\in B,$ and therefore $x\in \max
(B,R).$ Notice that $\max (B,R)$ can not be a singleton since $x$ and $y$
are different, therefore $c(B)=\max (B,R)\backslash \min (\max (B,R),L).$ I
will now show that $x\neq \min (\max (B,R),L).$ By hypothesis $x,y\in c(A)$
so $\max (A,R)$ is not a singleton, recall $x\neq y,$ which implies $%
c(A)=\max (A,R)\backslash \min (\max (A,R),L)$. Hence there exists some $z$
distinct from $x$ and $y$, such that $z=\min (\max (A,R),L)$. Hence $xLz.$
Since $z\in \max (A,R)$, $zRy$. Since $y$ belongs to the set $A,$ then, by
transitivity of $R,$ $zRb$ for all $b\in B$ which shows that $z\in \max
(B,R) $. Hence $x$ can not be the $L$-worst alternative in the shortlist $%
\max (B,R)$.

The following example shows that $c$ may fail $\alpha $. Let $c$ admit a
minimal compromise representation. Additionally, let $X=\{x,y,z\}$, with $R$
ranking all three alternatives indifferent and $xLyLz.$ The consequent MC
choices are given by the following table:%
\begin{equation*}
\begin{tabular}{ccccc}
$A$ & $\{x,y\}$ & $\{x,z\}$ & $\{y,z\}$ & $\{x,y,z\}$ \\ 
$c(A)$ & $\{x\}$ & $\{x\}$ & $\{y\}$ & $\{x,y\}$%
\end{tabular}%
\end{equation*}%
Note $y$ is chosen from the menu $\{x,y,z\},$ but not from the menu $\{x,y\}$
where it belongs, leading to a failure of $\alpha .$ This happens because
the vetoed alternative is different between the two menus: $y$ is the $L$%
-worst alternative in the shortlist $\{x,y\}=\max (\{x,y\},R),$ but $z$ is
the $L$-worst alternative in the shortlist $\{x,y,z\}=\max (\{x,y,z\},R).$%
\bigskip
\end{proof}

There is another sense in which the deviation of MC choice correspondences
from full rationality comes in the form of $\alpha $ failures. Consider the
following two axioms:\bigskip

\noindent $\gamma $: If $x\in c(A)\cap c(B)$, then $x\in c(A\cup B)$.\bigskip

\noindent No binary cycles (NBC): If $x\in c(\{x,y\})$ and $y\in c(\{y,z\})$%
, then $x\in c(\{x,z\})$.\bigskip

\noindent \textbf{Fact 2}: (Moulin, 1985) A choice correspondence satisfies
WARP if and only if it satisfies $\alpha ,\gamma $ and NBC.\bigskip

My next result indicates that MC choice correspondences satisfy $\gamma $
and NBC. Hence in view of the particular characterization given in Fact 2,
MC choice correspondences fail to be rational only because they may fail $%
\alpha $.

\begin{lemma}
Let $c$ admit a minimal compromise representation. Then $c$ satisfies $%
\gamma $ and NBC.
\end{lemma}

\begin{proof}
Suppose $c$ admits a minimal compromise representation. I will first show
that $c$ satisfies $\gamma $. Let $x\in c(A)\cap c(B)$. Then it must be the
case that $x\in \max (A,R)$ and $x\in \max (B,R),$ which means $x\in \max
(A\cup B,R).$There are two cases: $x$ is the only $R$-maximal element in
both menus or there exists an other $R$-maximal element, distinct from $x,$
in at least one menu. In the first case, $\{x\}=\max (A\cup B,R)$ and $%
c(A\cup B)=\max (A\cup B,R)=\{x\}$ and $x$ is chosen in $A\cup B$ as
desired. In the second case, there exists $y\not=x$ in, say, menu $A$ such
that $y\in \max (A,R)$. Since $yIx$, $y\in \max (A\cup B,R)$ as well. Since $%
x,y\in \max (A,R)$ it must be the case that $c(A)=\max (A,R)\backslash \min
(\max (A,R),L).$ Notice that there exists $z\in \max (A,R)$ such that $xLz,$
since $x\in c(A)$ and therefore $x\neq \min (\max (A,R),L).$ This last
statement means $zRx$ hence $z\in \max (A\cup B,R)$ and $x\neq \min (\max
(A\cup B,R),L)$ and $x\in c(A\cup B)$, as desired.

To show that $c$ satisfies NBC, let $x\in c(\{x,y\})$ and $y\in c(\{y,z\}).$
I have to show that $x\in c(\{x,z\}).$ Since $c$ has a minimal compromise
representation it must be the case that $x\in \max (\{x,y\},R)$ and $y\in
\max (\{y,z\},R),$ so $xRy$ and $yRz.$ Then, by transitivity of $R,$ $xRz$
which means $x\in \max (\{x,z\},R).$ If $\neg (z R x)$ then $c(\{x,z\})=\max
(\{x,z\},R)=\{x\}.$ If $zRx$. then $zRy$ and $yRx$ since $R$ is transitive.
Consequently $\max (\{x,y\},R)=\{x,y\},$ $\max (\{y,z\},R)=\{y,z\}$ and $%
\max (\{x,z\},R)=\{x,z\},$ hence $L$ vetoes an alternative in all three
doubleton menus. Recall $x\in c(\{x,y\})$ and $y\in c(\{y,z\}),$ then $xLy$
and $yLz$, which implies $xLz$ by transitivity of $L$. Hence $%
\{x\}=c(\{x,z\})$ once again. This finishes the proof.
\end{proof}

\section{Revelation of Preferences}

Suppose that $c$ admits a minimal compromise representation with the
underlying preferences $R$ and $L$. How does the resulting behavior reveal $%
R $ and $L$? The following example shows that there may be more than one way
to rationalize observed choices by this model.

\begin{example}
\emph{Consider the following choice correspondence:}%
\begin{equation*}
\begin{tabular}{ccccc}
$A$ & $\{x,y\}$ & $\{x,z\}$ & $\{y,z\}$ & $\{x,y,z\}$ \\ 
$c(A)$ & $\{y\}$ & $\{x\}$ & $\{y\}$ & $\{y\}$%
\end{tabular}%
\end{equation*}%
\emph{Note that this choice correspondence is decisive in every menu.
Furthermore it is rationalizable in the standard sense. One way to see that
this behavior admits a minimal compromise representation is to let }$L$\emph{%
\ be any linear order and take }$R$\emph{\ as follows: }$yPxPz$\emph{, where 
}$P$\emph{\ is the strict part of }$R$\emph{. Alternatively, if I take }$%
xIyPz$\emph{\ for }$R$\emph{, and }$zLyLx$\emph{, the resulting choices are
identical. }$\blacktriangle $
\end{example}

Yet, through violations of $\alpha $, a MC choice correspondence can quite
tractably reveal at least parts of the underlying weak order $R$ and the
linear order $L$. First, note that any chosen alternative must be at least
as good as any other feasible alternative in $R$. Hence $xRy$ if $x\in c(A)$
and $y\in A$ in some menu $A$. This mirrors the revelation of the underlying
preference in a standard rationality framework. However in my model, if $%
x\in c(A)$ and $y\in A\backslash c(A)$, this does not mean that $x$ is
strictly preferred to $y$. This is because $x$ and $y$ are perhaps
indifferent, but $y$ has been vetoed by the linear order $L$. However, any
such $y$ can easily be detected as their removal would impact behavior.
Indeed, if $y$ has been vetoed, then its removal from $A$ may lead to the
veto of the next alternative in $L$, meaning $c(A\backslash y)\not=c(A)$
even though $y\not\in c(A)$. If this happens, then I can also conclude that $%
yRa$ for all $a\in A$. Furthermore I can also conclude that $xLy$ for all $%
x\in c(A)$. However this revelation of $L$ is not complete. Suppose that $R$
shortlists only two alternatives $x$ and $y$ but $y$ is vetoed by $L$. In
this case the removal of $y$ has no impact on behavior as $c(A)=\{x\}$ since 
$y$ is vetoed, but $c(A\backslash y)=\{x\}$ as well, as only $x$ has been
shortlisted, stripping $L$ of its veto power.

The characterization exercise of the next section contains my main result,
Theorem 1, where these revelations play the key role. In it, I define the
binary relation $R$ as follows: $xRy$ iff $y$ belongs to a menu where the
removal of $x~$affects behavior. I also define a binary relation $L$ as
follows: $xLy$ iff $c$ is decisive in the menu $\left\{ x,y\right\} $ in
favor of $x$. I show that, under certain conditions which I will specify, $R$
is complete and transitive, i.e., a weak order. I also show that, again
under conditions, $L$ can be \ completed to a linear order. Furthermore, any 
$c$ satisfying the conditions I will identify behaves identically to a MC
choice correspondence defined by $R$ and $L$.

\section{Characterization}

In this section I\ will characterize the class of MC choice correspondences.
Let me begin with some notation which will help in the statement of two of
my conditions. For any choice correspondence $c$ and any menu $A$, let 
\begin{equation*}
r^{c}(A)=\{x\in A:c(A\backslash x)\not=c(A)\}.
\end{equation*}%
In words, $r^{c}(A)$ collects all members of $A$ whose removal impacts
behavior. Note $r^{c}$ is a choice correspondence itself: $r^{c}(A)\subseteq
A$ and, since $c(A)\subseteq r^{c}(A)$, $r^{c}(A)$ is nonempty. Clearly, the
removal of any chosen element will impact behavior. However the removal of
alternatives which are not chosen may also impact behavior. The following
observation indicates that this never happens if $c$ is rational.

\begin{lemma}
If a choice correspondence $c$ is rational, then $r^{c}(A)=c(A)$ for every
menu $A$.
\end{lemma}

\begin{proof}
Take a choice correspondence $c$. Suppose $c$ is rational and let $R$ be the
weak order which $c$ maximizes, i.e., $c(A)=\max (A,R)$ for all $A$. I will
show that $r^{c}=c$. By definition, $c(A)\subseteq r^{c}(A)$. Take any $x\in
A\backslash c(A)$. If $a\in \max (A,R)$, then $a\not=x$, $a\in A\backslash x$
and $a\in \max (A\backslash x,R)$ as well. Hence $c(A)\subseteq
c(A\backslash x)$. If $a\not\in \max (A,R)$ and $a\not=x$, on the other
hand, then there exists $a^{\prime }\not=x$ such that $a^{\prime }Pa$ and $%
a\not\in \max (A\backslash x,R)$, giving $c(A\backslash x)\subseteq c(A)$. I
conclude that $a\not\in r^{c}(A)$ and $c(A)=r^{c}(A)$.\bigskip
\end{proof}

The next example show that the statement in the preceding Lemma can not be
reversed, i.e., $r^{c}=c$ is not sufficient for the rationality of $c$.

\begin{example}
\emph{Consider the following choice correspondence on }$X=\{x,y,z\}$\emph{. }%
\begin{equation*}
\begin{tabular}{ccccc}
$A$ & $\{x,y\}$ & $\{x,z\}$ & $\{y,z\}$ & $\{x,y,z\}$ \\ 
$c(A)$ & $\{x,y\}$ & $\{z\}$ & $\{y\}$ & $\{x,y\}$%
\end{tabular}%
\end{equation*}%
\emph{Note }$r^{c}=c$\emph{\ but }$c$\emph{\ is not rational as it fails }$%
\alpha $\emph{: }$x\in c(\{x,y,z\})$\emph{\ but }$x\not\in c(\{x,z\})$\emph{%
. }$\blacktriangle $
\end{example}

I am now ready to state the characterizing conditions for MC choice
correspondences. \bigskip

\noindent \textbf{Condition 1.} If $x,y\in A\subset B,$ $x\in c(A),$ and $%
y\in c(B),$ then $x\in c(B).$\bigskip

This condition says that if an alternative $y$ is chosen in a menu, and $x$
is chosen in a smaller menu where $y$ is present, then $x$ must also be
chosen in the larger menu. This condition strengthens $\beta $ by weakening
its "if-part" as it does not insist that $y$ should be chosen in the smaller
menu $A$ for the conclusion to follow. The following example shows that the
strengthening is strict, as there are choice correspondences which satisfy $%
\beta $ but fail Condition 1.\bigskip 

\begin{example}
\emph{The following choice correspondence satisfies }$\beta $\emph{\ but
fails Condition 1.}%
\begin{equation*}
\begin{tabular}{ccccc}
$A$ & $\{x,y\}$ & $\{x,z\}$ & $\{y,z\}$ & $\{x,y,z\}$ \\ 
$c(A)$ & $\{x\}$ & $\{x\}$ & $\{y\}$ & $\{y\}$%
\end{tabular}%
\end{equation*}%
\emph{Note }$\beta $\emph{\ holds vacuously as }$c$\emph{\ is decisive in
every menu. However Condition 1 fails: }$x,y\in \{x,y\}\subset \{x,y,z\},$%
\emph{\ }$x\in c(\{x,y\}),$\emph{\ and }$y\in c(\{x,y,z\}),$\emph{\ but }$%
x\not\in c(\{x,y,z\}).$ $\blacktriangle $
\end{example}

The next condition says that Condition 1 should hold for the map $r^{c}$
associated with the choice correspondence $c$.\bigskip

\noindent \textbf{Condition 2. }For any $c$, $r^{c}$ satisfies Condition 1:
If $x,y\in A\subset B,$ $x\in r^{c}(A),$ and $y\in r^{c}(B),$ then $x\in
r^{c}(B).$\bigskip

Hence if the removal of $y$ changes behavior in menu $B$, and the removal of 
$x$ changes behavior in a smaller menu $A$ where $y$ belongs, then the
removal of $x$ should change behavior in menu $B$ as well. \bigskip

The next condition says that $c$ should not choose every feasible
alternative, except of course in singletons. \bigskip

\noindent \textbf{Condition 3. }For any nonsingleton menu $A$, there exists $%
x\in A$ such that $x\not\in c(A)$.\bigskip

Note that Condition 3 implies, in particular, that $c$ should be decisive in
doubleton menus. Furthermore, I\ have the following result.\bigskip

\begin{lemma}
If $c$ satisfies Conditions 1 and 3, then it must satisfy NBC as well.
\end{lemma}

\begin{proof}
Suppose $c$ satisfies Conditions 1 and 3, $x\in c(\{x,y\})$ and $y\in
c(\{y,z\})$ but $\{z\}=c(\{x,z\})$. Note, by Condition 3, this means $%
\{x\}=c(\{x,y\})$ and $\{y\}=c(\{x,y\})$. Consider the menu $\{x,y,z\}$.
Since all alternatives in $\{x,y,z\}$ have been chosen in a smaller menu,
Condition 1 implies that if one of them belongs to $c(\{x,y,z\})$, then all
do so. This violates Condition 3. \bigskip
\end{proof}

Next consider the following condition.\bigskip

\noindent \textbf{Condition 4. \ }If $x\in A\backslash c(A)$ and $x\in
c(A\cup \{y\}),$ then $y\notin c(A\cup \{y\}).$\bigskip

Imagine adding a new alternative $y$ to a menu $A$. Condition 4 says that
this can not lead to the inclusion of both $y$ and a previously unchosen
alternative $x$ in the choice set. If $x$ jumps in the choice set as a
result of the inclusion of $y$, then $y$ should not belong to the choice set.

Finally, my last condition is as follows.\bigskip

\noindent \textbf{Condition 5.} For all $A$, for all nonsingleton $%
B\subseteq r^{c}(A)$ there exists some $x\not\in c(B)$ such that $B=c(B)\cup
\{x\}$.\bigskip

Imagine choice in a menu of alternatives all of which impact choice in a
larger menu. Condition 5 says that all but one of these alternatives must
belong to the choice set.

I am now ready to state the main result.

\begin{theorem}
A choice correspondence admits a minimal compromise representation if and
only if it satisfies Conditions 1-5.
\end{theorem}

\begin{proof}
To begin, suppose that $c$ admits a minimal compromise representation and
let $R$ and $L$ be the underlying weak and linear orders.

To see that $c$ satisfies Condition 1, fix $x,y\in $ $A\subset B$ such that $%
x\in c(A)$ and $y\in c(B)$. I need to show that $x\in c(B).$ There is
nothing to show if $x=y$, so suppose $x\not=y$. By definition $x\in \max
(A,R)$ and $y\in \max (B,R)$, hence $x\in \max (B,R)$ and $y\in \max (A,R)$
as well. Since $\max (A,R)$ is not a singleton, there exists some $z\in \max
(A,R)\backslash c(A)$ such that $xLz$. Note $z\in \max (B,R)$ as well, so $%
x\not=\min (\{\max (B,R)\},L)$, giving $x\in c(B)$.

Next take $x,y\in $ $A\subset B$ such that $x\in r^{c}(A)$ and $y\in
r^{c}(B) $. To establish Condition 2, I need to show that $x\in r^{c}(B).$
There is nothing to show if $x=y$ or if $x\in c(B)$. Suppose $x\not=y$ and $%
x\not\in c(B)$. Note that minimal compromise representation implies $%
r^{c}(S)\subseteq \max (S,R)$ for any menu $S$. Hence $x\in \max (A,R)$, $%
y\in \max (B,R)$ and consequently $\max (A,R)\subseteq \max (B,R)$ and $x\in
\max (B,R)$ as well. This means $x$ is vetoed by $L$ in $B$. I need to show
that $c(B)$ contains at least two alternatives so that $x\in r^{c}(B)$.
Suppose $x\in c(A)$. Note $y\in \max (A,R)$ as well, so $\max (A,R)$
contains at least two alternatives. Hence an alternative must be vetoed in $%
A $, say some $a\not=x$. Then $a\in \max (A,R)\subseteq \max (B,R)$ and $xLa$%
, hence $x$ cannot vetoed in $B$, a contradiction. If $x\not\in c(A)$, on
the other hand, then $x$ is vetoed in $A$. However $x\in r^{c}(A)$, meaning $%
c(A) $ contains at least two distinct alternatives $a_{1},a_{2}\in \max
(A,R)\subseteq \max (B,R)$ such that $a_{i}Lx$. Then $a_{1},a_{2}\in c(B)$,
and the removal of $x$ from $B$ will result in a change in behavior, as I\
needed to show. I conclude that $x\in r^{c}(B)$.

To show that $c$ satisfies Condition 3, take a menu $A$ which is not a
singleton and suppose $c(A)=A$. Since $c(A)\subset \max (A,R)$ whenever $%
c(A) $ contains multiple alternatives, $A=c(A)\subset \max (A,R),$ and
impossibility.

Next, to show $c$ satisfies Condition 4, take $x,y$ and $A$ such that $x\in
A\backslash c(A)$ and $x\in c(A\cup \{y\}).$ I have to show that $y\notin
c(A\cup \{y\}).$ By definition of $r^{c}$, $x,y\in r^{c}(A\cup \{y\})$,
which means\textbf{\ }$x,y\in \max (A\cup \{y\},R)$ as well. Hence $x\in
\max (A,R)$ but for all $a\in c(A)$, $aLx$. Now take any $a\in \max (A\cup
\{y\},R)\backslash \{y\}$. Note such $a\in \max (A,R)$ as well. If $yLa$,
then \thinspace $yLx$ also and consequently $x$ is the $L$-worst alternative
in $\max (A\cup \{y\},R)$, meaning $x\not\in c(A\cup \{y\})$, a
contradiction. Hence $aLy$ for all $a\in \max (A\cup \{y\},R)\backslash
\{y\} $ and $y\not\in c(A\cup \{y\})$.

Finally, to see that $c$ satisfies Condition 5 take nonsingleton menus $A$
and $B$ such that $c(A)\not=c(A\backslash x)$ for all $x\in B$. Then $%
B\subseteq \max (A,R)$. Consequently $B=\max (B,R)$. Let $x=\min (B,L)$ so
that $c(B)=B\backslash \{x\}$, and Condition 5 follows.

In the reverse direction, suppose $c$ satisfies Conditions 1-5. Define $xRy$
iff there exists some menu $A$ such that $x\in r^{c}(A)$ and $y\in A$. Also
define $xLy$ iff $\{x\}=c(\{x,y\})$. I will now show that $R$ is a weak
order and $L$ is a linear order.

Completeness of $R$ follows from the definition of a choice correspondence.
Take any $x,y\in X$. If $x=y$, then $xRx$ since $\{x\}=c(\{x\})$. Otherwise
since $c(\{x,y\})\not=\varnothing $, $xRy$ or $yRx$ (or both). To see that $%
R $ is transitive, suppose $xRy$ and $yRz$. Then there exist menus $A$ and $%
B $ such that $x\in r^{c}(A)$ and $y\in r^{c}(B)$, $y\in A$ and $z\in B.$
Take any $w\in r^{c}(A\cup B)$. If $w\in A$, then $x\in r^{c}(A\cup B)$ by
Condition 2. If $w\in B$, then $y\in r^{c}(A\cup B)$ and therefore $x\in
r^{c}(A\cup B)$, again, by Condition 2. Since $z\in A\cup B$, $xRz$ as
desired. This proves that $R$ is a weak order.

To see that $L$ is complete, take any $x,y\in X$. If $x=y$, then $xLx$ as $%
\{x\}=c(\{x\})$. Otherwise Condition 3 implies $\{x\}=c(\{x,y\})$ or $%
\{y\}=c(\{x,y\})$. Hence $xLy$ or $yLx$ and $L$ is complete. Furthermore I
can not have $xLy$ and $yLx$ for distinct $x$ and $y$, hence $L$ is
asymmetric. The transitivity of $L$ follows from Lemma 4 and Condition 3 as
follows. If $xLy$, $yLz$, then Lemma 4 implies $x\in c(\{x,z\})$ and
Condition 3 implies $\{x\}=c(\{x,z\})$. Hence $xLz$ and $L$ is transitive.
This proves that $L$ is a linear order.

Now let $c_{R,L}$ be the MC choice correspondence defined by $R$ and $L$. I\
will show that $c=c_{R,L}$. First suppose that $x\in c(A)$. I need to show
that $x\in c_{R,L}(A)$. Since $c(A)\subseteq r^{c}(A)$, $x\in r^{c}(A)$. By
definition of $R$, then, $x\in \max (A,R)$. There are two cases to consider.
If $\max (A,R)=\{x\}$, then $\max (A,R)=\{x\}=c_{R,L}(A)$ by definition and $%
x\in c_{R,L}(A)$ as desired. Suppose now that $\max (A,R)$ contains multiple
alternatives. I need to show that $x$ is not the $L$-worst alternative in $%
\max (A,R)$. In other words, I need to find an alternative $a\in \max
(A,R)\backslash \{x\}$ such that $c(\{a,x\})=\{x\}$. Let $y\in \max
(A,R)\backslash x$. By Condition 3, $c(\{x,y\})$ is a singleton. If $%
c(\{x,y\})=\{x\}$, then $xLy$ and I\ am done. Suppose that $c(\{x,y\})=\{y\}$%
. Since $x\in c(A)$ and $\{x,y\}\subset A$, as I keep adding alternatives to
menu $\{x,y\}$ to reach menu $A$, $x$ must jump in the choice set at some
point. In other words, there must exist a menu $D$ and an alternative $z\in
A\backslash D$ such that $\{x,y\}\subseteq D\subseteq A$, $x\not\in c(D)$
and $x\in c(D\cup \{z\})$. (If no such $D$ and $z$ exist, then $x\not\in
c(A) $.) Note that $z\in r^{c}(D\cup \{z\})$ and therefore $zRx$. Hence $%
z\in \max (A,R)$ as well. However $z\not\in c(D\cup \{z\})$ by Condition 4.
Now consider the menu $\{x,z\}$. If $\{z\}=c(\{x,z\})$, then Condition 1
dictates that $z\in c(D\cup \{z\})$, a contradiction. Then, by Condition 3, $%
c(\{x,z\})=\{x\}$ and $xLz$. This proves that $x$ is not the $L$-worst in $%
\max (A,R)$. I conclude that $x\in c_{R,L}(A)$.

To finish, take $x\in c_{R,L}(A)$. I need to show that $x\in c(A)$. By
definition $x\in \max (A,R)$. If $\max (A,R)=\{x\}$, there exists no $y\in
A\backslash \{x\}$ such that $yRx$. This implies that $c(A)$ can only
contain $x$. Since $c(A)$ is nonempty, it \emph{has to} contain $x$, as
desired. Now suppose $\max (A,R)\not=\{x\}$, i.e., that there exists some $%
z\not=x$ such that $\{x,z\}\subseteq \max (A,R)$. Since $x\in c_{R,L}(A)$
and $\max (A,R)$ contains multiple alternatives, $x$ is not the $L$-worst
alternative in $\max (A,R)$. Hence I can take $z$ such that $xLz$, i.e., $%
c(\{x,z\})=\{x\}$.

Towards a contradiction suppose $x\not\in c(A)$ and pick $y\in c(A)$. If $%
x\in c(\{x,y\})$, then by Condition 1, $x\in c(A)$ as well, a contradiction.
By Condition 3, then $\{y\}=c(\{x,y\})$.

Since $zRx$ and $zRy$, there exist menus $B_{x}$ and $B_{y}$ such that $x\in
B_{x}$, $y\in B_{y}$ and $z\in r^{c}(B_{x})\cap r^{c}(B_{y})$. By Condition
2, then $z\in r^{c}(B_{x}\cup B_{y})$. Since $c(\{x,z\})=\{x\}$, $x\in
r^{c}(\{x,z\})$ and Condition 2 implies $x\in r^{c}(B_{x}\cup B_{y})$.
Similarly, since $c(\{x,y\})=\{y\}$, $y\in r^{c}(B_{x}\cup B_{y})$.

Now I will use Condition 5. Consider the menu $\{x,y,z\}\subseteq
r^{c}(B_{x}\cup B_{y})$. Conditions 3 and 5 imply that $c(\{x,y,z\})$
contains exactly two alternatives. If $c(\{x,y,z\})=\{y,z\}$, then Condition
1 fails since $c(\{x,z\})=\{x\}$. Similarly if $c(\{x,y,z\})=\{x,z\}$, as $%
\{y\}=c(\{x,y\})$. Hence $c(\{x,y,z\})=\{x,y\}$. Now Condition 1 implies $%
x\in c(A)$, as $y\in c(A)$. This contradiction finishes the proof.
\end{proof}

\section{Conclusion}

In this work I study a two-stage choice correspondence defined by a weak
order $R$ and a linear order $L$. In any menu $A$, first $R$ shortlists its
maximal alternatives and next $L$ vetoes its worst alternative in the
shortlist, provided that the shortlist contains multiple alternatives. Hence
the behavior features a compromise from the maximization of $R$, but this
compromise is minimal, as $L$ can only veto its least preferred alternative.
Only if $R$ shortlists two candidates does $L$ make the choice. If three or
more candidates are shortlisted, then $L$ can only veto a single one of
them. Moreover if only a single alternative is shortlisted, $L$ has no
effect on behavior.

I show that this model satisfies various rationality conditions, most
notably $\beta $, $\gamma $ and no-binary-cycles. However it may fail the
famous $\alpha $ axiom. I provide five novel conditions which together
characterize this model. I\ leave for future work the generalization to the
scenario where the second preference is also a weak order, and could veto
multiple alternatives. 


\bigskip

\bigskip \bigskip

\begin{thebibliography}{9}
\bibitem{} Bajraj, G. and \"{U}lk\"{u}, L., (2015). Choosing two Finalists
and the Winner. Social Choice and Welfare, December 2015, Volume 45, Issue
4, pages 729-744

\bibitem{} Manzini, P. and Mariotti, M. (2007). Sequentially Rationalizable
Choice. American Economic Review, 97(5), 1824-1839

\bibitem{} Moulin, H. (1985). Choice functions over a finite set: A summary.
Social Choice and Welfare, 2(2), 147-160

\bibitem{} Horan, S. (2016). A Simple Model of Two-Stage Choice.Journal of
Economic Theory, Issue C, 372-406

\bibitem{} Garc\'{\i}a-Sanz, Mar\'{\i}a D. \& Alcantud, Jos\'{e} Carlos R
(2015). Sequential rationalization of multivalued choice, Mathematical
Social Sciences, Elsevier, vol. 74(C), pages 29-33
\end{thebibliography}
\end{document}